\documentclass{article}

\usepackage{array}

\usepackage{arxiv}
\usepackage{comment}        
\usepackage[utf8]{inputenc} 
\usepackage[T1]{fontenc}    
\usepackage{hyperref}       
\usepackage{url}            
\usepackage{booktabs}       
\usepackage{amsfonts}       
\usepackage{bbold}          
\usepackage{bbm}
\usepackage{amsmath}        
\usepackage{amsthm}         
\usepackage{stmaryrd}       
\usepackage{nicefrac}       
\usepackage{microtype}      
\usepackage{lipsum}		    
\usepackage{graphicx}
\usepackage{natbib}
\usepackage{doi}
\usepackage{amsmath}
\usepackage{mathtools}      
\usepackage{float}
\usepackage[dvipsnames]{xcolor}
\usepackage{enumitem} 
\usepackage{multirow}
\usepackage{breqn} 

\newtheorem{theorem}{Theorem}[section] 

\title{Credit Spreads' Term Structure: Stochastic Modeling with CIR++ Intensity}

\author{ {Mohamed Ben~Alaya} \\
	LMRS, UMR 6085 CNRS\\ 
	Universit\'e de Rouen Normandie\\
	1 Rue Thomas Becket, 76130 Mont-Saint-Aignan\\
	\text{mohamed.ben-alaya@univ-rouen.fr} 
	\And
    {Ahmed~Kebaier} \\ 
	LAMME, UMR 8071\\
	Universit\'e d'Evry\\
	23 Bd. de France, 91037 \'Evry Cedex\\
	\text{ahmed.kebaier@univ-evry.fr} 
 	\And
    {Djibril~Sarr} \\
	LAGA, UMR  7539 CNRS\\
	Universit\'e Sorbonne Paris Nord\\
	99 Av. Jean Baptiste Cl\'ement, 93430 Villetaneuse\\
	FBH Associ\'es \\
	11 Rue du 4 septembre, 75002 Paris\\
	\text{sarr@math.univ-paris13.fr} \\
	\text{djibril.sarr@fbh-associes.com} 
}

\date{}

\renewcommand{\headeright}{}
\renewcommand{\undertitle}{}
\renewcommand{\shorttitle}{Credit Spreads' Term Structure: Stochastic Modeling with CIR++ Intensity}

\hypersetup{
pdftitle={Credit Spreads' Term Structure: Stochastic Modeling with CIR++ Intensity},
pdfsubject={Spreads Models, Risk Neutral, Financial modeling},
pdfauthor={Djibril SARR},
pdfkeywords={Credit spreads, CIR, Calibration, Neural Networks},
}

\begin{document}
\maketitle

\begin{abstract}

This paper introduces a novel stochastic model for credit spreads. The stochastic approach leverages the diffusion of default intensities via a CIR++ model and is formulated within a risk-neutral probability space. Our research primarily addresses two gaps in the literature. The first is the lack of credit spread models founded on a stochastic basis that enables continuous modeling, as many existing models rely on factorial assumptions. The second is the limited availability of models that directly yield a term structure of credit spreads. An intermediate result of our model is the provision of a term structure for the prices of defaultable bonds. We present the model alongside an innovative, practical, and conservative calibration approach that minimizes the error between historical and theoretical volatilities of default intensities. We demonstrate the robustness of both the model and its calibration process by comparing its behavior to historical credit spread values. Our findings indicate that the model not only produces realistic credit spread term structure curves but also exhibits consistent diffusion over time. Additionally, the model accurately fits the initial term structure of implied survival probabilities and provides an analytical expression for the credit spread of any given maturity at any future time.

\end{abstract}

\keywords{Credit spread, Defaultable bond, CIR++ intensity, Quantitative finance, Risk measurement}

\section{Introduction}
\label{sec:SOA}
Credit spreads are a key indicator in finance. They are a versatile tool that effectively serve different purposes. Their most common use is probably as a creditworthiness measure of counterparties, as high credit spreads attributed to an entity (corporates, sovereigns, \dots) indicate that the market perceives a high credit risk, while low values suggest lower risks. Credit spreads are also required when pricing most market instruments. For instance, discount factors, $d(t,T)$ at time $t$ and maturity $T$ are function of credit spreads, $\operatorname{Sp}(t,T)$ and zero-coupon rates, $Z(t,T)$, as in its simplest expression (no adjustment spread for instance), we can write $d(t,T)=e^{-t\left[Z(t,T)+\operatorname{Sp}(t,T) \right]}$. This implies that computing net present values for most products (securities, bonds, interest rate swaps, \dots) requires credit spreads. Therefore, when performing full-revaluation risk computations, practitioners require models that can accurately propagate credit spreads. Similarly, regulatory risk computations like Credit Valuation Adjustment (CVA) often require credit spreads (see, for example, \citep{basel_cva}). In a less common use, they are also known to be macroeconomic indicators that can inform the state of an economy \citep{imf_spread, ecb_spread}. 

For these reasons, financial institutions need to have reliable credit spread models. The purpose of this document is to introduce a new stochastic model for credit spreads. The model diffuses default intensities through a CIR++ (shifted Cox-Ingersoll-Ross) intensity model and, through an expression that is obtained in section \ref{sec:RNPmodel} computes the credit spread. Our model's interest is first in the fact that it is derived from a classic stochastic model and therefore retains all the desirable properties from stochastic differential equation (SDE) driven models like most interest rate models. Indeed, the model natively provides a term structure of credit spreads, and it fits at $t=0$ the term structure of the market implied survival probability curve. Also, the credit spread formula in the model is a closed formula. As we will discuss in section \ref{sec:calib}, the model can be easily calibrated as, from the SDE describing the default intensity or the equation deriving the spreads expression, one can infer different types of equations describing implied volatilities, instruments prices and so on to calibrate to. It produces realistic spreads curves dynamics that are similar to market observations. 

The analytical development we conducted to obtain an explicit expression of credit spreads also yielded another very interesting result: an analytical expression for the price of a defaultable bond. Understanding the pricing of defaultable bonds is crucial for both investors and financial institutions, as it directly impacts the assessment of credit risk and the management of investment portfolios. Defaultable bonds, unlike risk-free bonds, carry the possibility of issuer default, necessitating sophisticated pricing models to accurately reflect their risk and return characteristics. The practical implications of defaultable bond pricing are vast. Accurate pricing models enable market practitioners to price and manage bond portfolios effectively, mitigating potential losses due to credit events. Furthermore, they support the development of hedging strategies that protect against the financial impacts of defaults. Additionally, these models aid in stress-testing financial institutions' capital positions under adverse economic conditions. For instance, a model that values defaultable bonds based on macroeconomic variables such as foreign exchange rates can reveal how an institution's capital might be affected during market stress \citep{lo2000stress}. 

Various sophisticated models have been developed to price defaultable bonds. Early work by \citep{Duffie1999} is part of the reduced-form models family, also known as intensity-based models, to which our work also belongs. These models treat default as a random process with a specified hazard rate or intensity. They do not explicitly model the firm's assets but instead focus on the likelihood of default over time. Intensity-based models are particularly useful for capturing the term structure of defaultable bonds. They are often contrasted with structural models, such as the Merton model, which are based on the firm's asset value and its volatility. In structural models, the default event is typically modeled as occurring when the firm's asset value falls below a certain threshold (default barrier). This type of model was used by \citep{Chen2009} to price defaultable bonds. More recently, \citep{Russo2020} found closed-form solutions for pricing bonds under a two-factor Gaussian model. Structural models provide a different perspective by focusing on the firm's financial health and its impact on default risk. 

The interest of our work regarding the pricing of defaultable bonds is twofold. First, the approach we develop does not introduce any complexity in the practical use of the model; the expression of the bond price is analytical and the calibration procedure is also provided, as it is the same as the one that will be introduced for the credit spread model. Second, to our knowledge, this is the first analytical expression of the defaultable bond price under a CIR++ intensity framework.

The primary challenge practitioners encounter when attempting to employ a credit spread model is the absence of a widely used and popular model. Unlike for interest rate modeling where models such as the Hull-White model \citep{Hull1993}, its two-factor extension \citep{Hull1994}, the Cox-Ingersoll-Ross model \citep{cox1985theory}, or the historical Vasicek model \citep{Vasicek1977} along with a list of other models are known to be very popular and were the subject of extensive studies analyzing their behavior and their performance (see for instance, \cite{Chan1992}), there does not seem to be one or many commonly used models for credit spreads. What made these models popular was their ease of use and effectiveness; our aim is to achieve the same: simplicity and efficiency.

The literature is rich with work on how to understand or forecast credit spreads. Most of the models in the literature seem to be factorial, meaning that they aim at explaining the variations of credit spreads via other variables. Such models have many advantages: they are, at least in the modeling part, easier to manipulate and use because, in general, the models are linear, their calibration is, in general, automatic (from factorial regressions or autoregressive models), and the interpretation of the model is more convenient as we deal with real meaningful indicators. However, they have drawbacks that motivate, for us, the use of SDE driven model: in general they do not allow native term structures, they do not allow continuous-time dynamics like SDEs, conducting stress tests is less flexible as we have with SDEs many native features (for instance, stress testing via model parameters, typically model volatility, introduction of jumps or RW stress). It is worth noting that the native term structure in SDE-driven models allows for the continuous monitoring of the indicator being considered, over different maturities, providing a more granular view. This feature is particularly useful when entities need to evaluate potential shocks or movements at various points on said indicator. Therefore, it is a common requirement in risk management as well as in pricing issues.

\citep{spreadsdeterminant} for instance, conduct a study using 85 years of corporate bonds data to find determinants of credit spreads. The author uses a self-extracting threshold model with the inflation characterizing the threshold, and the factors used are the levels series for the spread, US treasury bill (T-bill), equity and industrial production variables. \citep{spreadseurobond} also uses autoregressive models that test both ARCH and GARCH models. The author uses FTSE returns, the term spread, the US dollar mark to the sterling pound and previous credit spreads values to describe CS. In \citep{changeinCS}, the authors explain 67\% of the variation of credit spreads using two types of indicators, common factors, they are used by all entities and include equity market return, changes in 5Y government rates, and company-level factors that include stock momentum or change in equity volatility. An interesting feature of their process is that they include in the set of common variables a dummy variable that is 1 when the US Federal funds rate increases and 0 otherwise. They find that this variable is significant only for the highest-rated counterparties where expansionary Fed policies (i.e., decreases in the rate value) reduce credit spreads. 

One reason for the popularity of factorial model is probably the so-called "credit spread puzzle", which states that structural approach to credit risk and corporate bond pricing (like Merton, \citep{mertonspread}) underestimates the credit spread for investment-grade counterparties. In \citep{mythbuster}, the author introduces a new calibration methodology to use the Black-Cox model leading to more precise estimates of investment grade default probabilities. One key feature of the authors work is that instead of relying solely on the historical default rate for a specific maturity and credit rating as a proxy for default probability at that exact maturity and rating, they opt to utilize a wide cross section of default rates at different maturities and ratings. The primary characteristic of their approach that enables them to consolidate default rate data across various credit ratings and maturity periods is the underlying assumption that companies, regardless of their credit ratings and the maturity of their bonds, will still adhere to a common default threshold or boundary. An emphasis will therefore be put on back-testing the model and its calibration.  

Additionally, authors have introduced tree-based models for credit spreads or credit spread-related derivatives pricing (see for example, \citep{tree1, tree2, tree3}). 

Like in this work, in \citep{factorhazard}, the authors also use default intensity to obtain credit spreads. However, the two approaches are quite different. We do not require in our model an explicit derivation of the default intensity as the model is based on the assumption that it can be described by a CIR++ model (just as we assume that instantaneous rates can be described by G2++,  Vasicek, \dots to obtain zero-coupon term structures). In the authors' work, the default intensity is expressed as a function of interest rates and market value of the assets of the firm. This requires the practitioner to hold information that is neither easily available for all types of counterparties (it might be accessible for large companies and some sovereigns, but not for most sub-sovereigns/regional governments and local authorities, small corporates, \dots) nor easily proxied. It also requires having a model for interest rates and the market value of companies' assets.
To our knowledge, this is the first modeling of credit spreads using default intensities described by a CIR++ model. 

In what follows, we will start by deriving in Section \ref{sec:BondPricing} the risky bond price. Then, in Section \ref{sec:RNPmodel} we obtain the expression for the credit spread under the risk neutral measure $\mathbb{Q}$. Section \ref{sec:calib} then focuses on the calibration of the model. We discuss three calibration approaches before exploring in depth the one more suited for risk management. Section \ref{sec:mcrnp} illustrates the propagation of credit spreads using the model. Finally in Section \ref{sec:backtest} we perform the back-testing of the credit spread model. 

\section{Defaultable Bond Pricing under the CIR++ Intensity}
\label{sec:BondPricing}
The aim of this work is to provide practitioners with a stochastic model able to reliably reproduce and propagate the term structure of credit spreads. The proposed approach consists of two main steps:
\begin{itemize}
    \item Modeling default intensities using the CIR++. 
    \item Deriving the corresponding expression that links default intensities to the term structure of credit spreads. 
\end{itemize}
The search for an analytical expression for credit spreads led to the derivation of a closed-form expression for the prices of defaultable bonds.

\paragraph{Model setting}\mbox{}\\
In everything that follows, let $\tau$ be the instant of default and $T>0$ the time horizon. The risk neutral model is set in a filtered probability space $\left(\Omega, \mathcal{F}, (\mathcal{F}_t)_{t \in[0,T]}, \mathbb{Q} \right)$, $(\mathcal{F}_t)_{t \in[0,T]}$ is the natural filtration of the default-less market, $\mathbb{Q}$ defines the risk neutral probability space and $Q(\cdot)$ the probability measure. Let $r(t)=r_t$ be the instantaneous risk-free rate of the market and $\lambda(t)=\lambda_t$ the default intensity. The default intensity process $(\lambda_t)_{t\in[0,T]}$ and the interest rates process $(r_t)_{t\in[0,T]}$ are $\mathcal{F}_t$-measurable. They are also taken to be $\mathcal{F}_t$-conditionally independent (see hypotheses \eqref{hyp:H3} and \eqref{hyp:H4}). Let $\mathcal{G}_t$ be the continuously enlarged filtration that includes default information, $\mathcal{G}_t=\mathcal{F}_t\vee\sigma(\{\tau<s\},\;s\leq t)$. This setup is similar to others from the literature (see for example, \cite{chiarella2011modelling, jeanblanccir}).
\paragraph{Density Hypothesis and Immersion property}\mbox{}\\
We assume that for any time $t \geq 0$, the density of $\tau$ conditionally to $\mathcal{F}_t$ exists and is defined at any time $t$ by an $\mathcal{F}_t \otimes \mathcal{B}(\mathbb{R}^+)$- measurable function $(\omega, \tilde{t}) \longrightarrow \alpha_t(\tilde{t})$. Consequently, for $f$, a bounded Borel function: 
$$
\mathbb{E}\left[f(\tau) \mid \mathcal{F}_t\right]=\int_0^{\infty} f(u) \alpha_t(u) \mathrm{d} u, \quad \text { a.s. }.
$$
This assumption is known as the \textit{density hypothesis}, see for instance \cite{jiao}. We consider the particular case in which the reference filtration, observed after the default time, does not provide any additional insight into the conditional distribution of the default. This translates into the following equality for the conditional density function:
\begin{equation}
    \label{hyp:Hd}
    \tag{H1}
    \alpha_t(\tilde{t})=\alpha_{\tilde{t}}(\tilde{t}), \quad \forall t \geq \tilde{t} \quad \text{ a.s. },
\end{equation}
Furthermore, it is well known that under hypothesis \eqref{hyp:Hd}, the \textit{immersion property}, also known as the \textit{H-hypothesis} holds. That is, for any fixed t and any bounded $\mathcal{G}_t$-random variable $Y_t$, 
\begin{equation}
    \label{hyp:H2}
    \tag{H2}
    \mathbb{E}\left[Y_t \mid \mathcal{F}_{\infty}\right]=\mathbb{E}\left[Y_t \mid \mathcal{F}_t\right] \quad \text { a.s. },
\end{equation}
See for instance equations (11) and (12) in \cite{jiao} where a similar setup to the one described in this paragraph is presented.
\paragraph{Independence hypotheses}\mbox{}\\
We make the following hypotheses regarding the independence between interest rates and the moment of default $\tau$:
\begin{equation}
    \label{hyp:H3}
    \tag{H3}
    \exp \left(-\int_t^T r_s d s\right) \text{ is } \mathcal{F}_t\text{-conditionally independent of } \mathbbm{1}_{\{\tau>T\}},
\end{equation}
\begin{equation}
    \label{hyp:H4}
    \tag{H4}
    \exp\left(-\int_t^T r_s d s\right) \text{ is $\mathcal{F}_t$-conditionally independent of } \alpha_T(u), \; u>t.
\end{equation}

\paragraph{Default Intensity}\mbox{}\\
The default intensity, also called in the literature hazard rate, $\lambda(t)$ verifies:
\begin{equation}
    \label{eq:di_CSPRN}
    \lambda(t) dt = Q\left(\tau \in[t, t+d t] \mid \tau>t, \mathcal \mathcal{F}_t\right).
\end{equation}
Equation\eqref{eq:di_CSPRN} means that $\lambda(t)dt$ is the probability of default in $[t,\,t+dt]$. The result is derived from intensity models or reduced-form models that are in finance models that do not try to explain reason behind defaults (unlike factorial models, for instance). In such models, we often assume that default events correspond to jumps of inhomogeneous Poisson process where the previous relation is verified. That is to say, $\tau = \operatorname{inf}\{ t\geq 0, \Lambda(0,t)\geq \Theta \}$ where $\Theta$ is a unit exponentially distributed random variable. We also introduce the cumulative hazard rate that is defined by: 
\begin{equation}
    \label{eq:cumhr_CSPRN}
    \Lambda(0,t) = \int_0^t{\lambda_u du}.
\end{equation}

\paragraph{CIR++ Model}\mbox{}\\
The dynamic of the CIR++ model is given by:
\begin{equation}
    \label{eq:cirrnp_CSPRN}
    \begin{cases}
        & \lambda(t) = y(t) + \psi(t)  \\
        & dy(t) = \kappa\left(\theta-y(t)\right)dt + \sigma\sqrt{y(t)}dW_t,
    \end{cases} 
\end{equation}
where $(W_t)_{t\in[0,T]}$ is $\mathbb{Q}$-Brownian motion, $\kappa, \theta, \sigma>0$. The Feller condition, $2\kappa \theta \geq \sigma^2$ and $y(0) = y_0>0$ guarantees $y(t)=y_t>0$ under $\mathbb{Q}$. $\psi(t)$ is the deterministic function that we will use to fit the initial market-implied\footnote{We present two solutions to imply the survival probabilities, the more common solution using CDS prices (Section \ref{sec:calib}) and a solution specific to our model using credit spreads (Section \ref{sec:mcrnp})} survival probabilities, it will be explicited later on. Using a CIR-type model to diffuse default intensity is quite convenient, as we can easily ensure positive values. It actually is a common practice (see for instance, \citep{cirint} or more recently, \citep{cirint2}). When using the CIR model to diffuse default intensity, it is called the CIR intensity model. 
Many important works have been done regarding the relations between the CIR++ intensity and the classic CIR++ model for interest rates. The main result we will use is the equivalence between the survival probability, $S(\cdot,\cdot)$ in an intensity model and the zero-coupon bond price $P(\cdot,\cdot)$ in short rate models. This is because on one hand, as we already know about interest rates models, $P(0, t)=\mathbb{E}^{\mathbb{Q}}\left(e^{-\int_0^t r(u) d u}\right)$. On the other hand, it is quite straightforward to show that $S(0,t)=Q(\tau>t)=\mathbb{E}\left[e^{-\int_0^t \lambda(u) d u}\right]$. This means that both quantities can be written as Laplace transform of CIR-type processes. More detailed explanations on the meaning of the equivalence and the way to prove it can be found in sections 21.1.1.1 and 22.7.2 of \citep{brigo2006interest}. The price of the risk-free zero-coupon bond is:
\begin{equation}
    \label{eq:riskfreebond}
    P(t, T)=\mathbb{E}^{\mathbb{Q}}\left[\exp \left(-\int_t^T r_s d s\right) \bigg\vert \mathcal{F}_t\right],
\end{equation}
and the survival probability is given by:
\begin{equation}
    \label{eq:survprob}
    S(t, T)=\mathbb{E}^{\mathbb{Q}}\left[\exp \left(-\int_t^T \lambda_s d s\right) \bigg\vert \mathcal{F}_t\right].
\end{equation}

\paragraph{Defaultable bond price}\mbox{}\\
We recall that in our setup we consider, $(r_t)_{0\le t\le T}$ to be a $(\mathcal{F}_t)_{0\le t\le T}$-measurable process that is $\mathcal{F}_t$-conditionally independent of the default intensity, and therefore $\tau$-defined process.  
Then, the fair-price of a risky-bond (defaultable bond) with recovery rate $\delta\in(0,1)$ is
\begin{equation*}
    H(t, T)=\mathbb{E}^{\mathbb{Q}}\left[\exp \left(-\int_t^T r_s d s\right)\left(\delta \mathbbm{1}_{\{\tau \leq T\}}+\mathbbm{1}_{\{\tau>T\}}\right) \bigg\vert \mathcal{G}_t\right].
\end{equation*}
It is important to condition with respect to $\mathcal{G}_t$, as the bond's issuer may default.

\begin{theorem}
    \label{theo:ExpDefBond}
    Under the previous notations, and under hypotheses (H1-4), the price of the defaultable bond price of maturity $T$ at time $t$, $H(t,T)$ as a function of the risk-free zero-coupon bond price $P(t,T)$ is given by: 
    \begin{equation}
    \label{eq:defbondfinal}
        \begin{dcases}
            H(t, T) = P(t,T)\left[\delta+(1-\delta) \frac{S^m(0,T)}{S^m(0,t)}\frac{A(0,t)}{A(0,T)}\frac{e^{-B(0,t)y_0}}{e^{-B(0,T)y_0}}A(t,T)e^{-B(t,T)[\lambda(t)-\psi(t)]} \right]\\
            \psi(t) = \lambda^m(t)+D(t)-y_0E(t), 
        \end{dcases}
    \end{equation}
    where, 
    \begin{equation*}
    \begin{aligned}
        A(t,T) &= \left( \frac{2he^{\frac{1}{2}(\kappa+h)(T-t)}}{2h+(\kappa+h)\left( e^{(T-t)h} - 1 \right )} \right)^{\frac{2\kappa\theta}{\sigma^2}}\text {, }
        B(t,T) = \frac{2\left( e^{(T-t)h} - 1 \right )}{2h+(\kappa+h)\left( e^{(T-t)h} - 1 \right )}\\
        h &= \sqrt{\kappa^2+2\sigma^2}\\
        D(t) &= \frac{d}{dt}ln(A(0,t))\text{, }E(t) = \frac{d}{dt}B(0,t),
    \end{aligned}
\end{equation*}
$S^m(0,t)$ is the initial market's implied survival probability, for horizon $t$ and $\lambda^m(t)$ denotes the initial market default intensity for the horizon $t$.
\end{theorem}

Note that both $\lambda^m(t)$ and $S^m(0,t)$ are usually inferred from CDS prices. In the paragraph \textit{Data Set} of Subsection \ref{sec:calib} and in the paragraph \textit{Fitting the initial market term structure} of Subsection \ref{sec:mcrnp}, we provide more details on how they are obtained.

\begin{proof}
At first, we write
\begin{equation}
    \label{eq:risky0}
    \begin{aligned}
        &\begin{aligned}
        H(t, T) 
        & =\delta \mathbb{E}^{\mathbb{Q}}\left[\exp \left(-\int_t^T r_s d s\right) \mathbbm{1}_{\{\tau \leq T\}} \bigg\vert \mathcal{G}_t\right]+\mathbb{E}^{\mathbb{Q}}\left[\exp \left(-\int_t^T r_s d s\right) \mathbbm{1}_{\{\tau>T\}} \bigg\vert \mathcal{G}_t\right] =: H_1+H_2. \\
        \end{aligned}
    \end{aligned}
\end{equation}
For $H_1$, we have 
\begin{eqnarray}
 \label{eq:risky1}
    H_1 
     &=&\delta \mathbb{E}^{\mathbb{Q}}\left[\exp \left(-\int_t^T r_s d s\right) \bigg\vert \mathcal{G}_t\right]-\delta \mathbb{E}^{\mathbb{Q}}\left[\exp \left(-\int_t^T r_s d s\right)\left(\mathbbm{1}_{\{\tau>T\}}\right) \bigg\vert \mathcal{G}_t\right]\notag\\
      &=&\delta \mathbb{E}^{\mathbb{Q}}\left[\exp \left(-\int_t^T r_s d s\right) \bigg\vert \mathcal{G}_t\right]-\delta H_2.
\end{eqnarray}
Using Theorem 3.1 of \cite{jiao}, we have:
\begin{equation*}
    \begin{aligned}
        H_{11} &= \mathbb{E}^{\mathbb{Q}}\left[\exp \left(-\int_t^T r_s d s\right) \bigg\vert \mathcal{G}_t\right]\\
        &=H^{bd}_{{11}_t}\mathbbm{1}_{\{t<\tau\}}+H^{ad}_{{11}_t}\mathbbm{1}_{\{\tau\leq t\}}.
    \end{aligned}
\end{equation*}
$H^{bd}_{{11}_t}$ and $H^{ad}_{{11}_t}$ respectively define the before default and after default components. Under hypothesis \eqref{hyp:H2}, the immersion property, we have (see Subsection 3.1 of \cite{jiao}), 
\begin{equation}
    \label{eq:Had}
    H^{ad}_{{11}_t} = \mathbb{E}^{\mathbb{Q}}\left[\exp \left(-\int_t^T r_s d s\right) \bigg\vert \mathcal{F}_t\right] = P(t,T).
\end{equation}
Then for the component before default we have: 
\begin{equation}
    \label{eq:Hbd}
    H^{bd}_{{11}_t}=\frac{\mathbb{E}\left[\int_t^{\infty} {\exp\left(-\int_t^T r_s d s\right) \alpha_T(u) du} \bigg\vert \mathcal{F}_t\right]}{Q\left(\tau>t \mid F_t\right)}.
\end{equation}
From hypothesis \eqref{hyp:H4} we write:
\begin{equation*}
    \begin{aligned}
    H^{bd}_{{11}_t}& =\frac{\mathbb{E}\left[\int_t^{\infty} {\exp\left(-\int_t^T r_s d s\right) \alpha_T(u) du} \bigg\vert \mathcal{F}_t\right]}{Q\left(\tau>t \mid F_t\right)}\\
    &=\frac{\mathbb{E}\left[\exp\left({-\int_t^T r_s d s}\right) \bigg \vert \mathcal{F}_t\right] \times \mathbb{E}\left[\int_t^{\infty} \alpha_T(u) du  \bigg \vert \mathcal{F}_t\right]}{Q\left(\tau>t \mid \mathcal{F}_t\right)} \\
    &=\frac{P(t, T) \times \mathbb{E}\left[Q\left(\tau>t \mid \mathcal{F}_T\right) \mid \mathcal{F}_t\right]}{Q\left(\tau>t \mid \mathcal{F}_t\right)}\\
    \end{aligned}
\end{equation*}
\begin{equation}
    \label{eq:Hbd}
    \hspace{-7.cm}H^{bd}_{{11}_t}=P(t,T).
\end{equation}
Therefore from Equations \eqref{eq:Had} and \eqref{eq:Hbd}, it follows that $H_{11} = \mathbb{E}^{\mathbb{Q}}\left[\exp \left(-\int_t^T r_s d s\right) \bigg\vert \mathcal{G}_t\right]=P(t,T)\mathbbm{1}_{\{t<\tau\}}+P(t,T)\mathbbm{1}_{\{\tau\leq t\}}=P(t,T)$.
We then have $H(t,T)=\delta P(t, T)+(1-\delta)H_2$.
For $H_2$ we  recall that from the switching filtration theorem (see e.g. \cite[Section 22.5]{brigo2006interest}), for any $\mathcal{G}_{T}$-measurable $Z$ we have
\begin{equation*}
\mathbb{E}\left(\mathbbm{1}_{\{\tau>T\}} Z \mid \mathcal{G}_t\right)=\frac{\mathbbm{1}_{\{\tau>t\}}}{Q\left(\tau>t \mid \mathcal{F}_t\right)} \mathbb{E}\left(\mathbbm{1}_{\{\tau>T\}} Z \mid \mathcal{F}_t\right).
\end{equation*}
Using this result on $H_2$ leads to:
\begin{equation*}
H_2=\mathbb{E}^{\mathbb{Q}}\left[\exp \left(-\int_t^T r_s d s\right) \mathbbm{1}_{\{\tau>T\}} \bigg\vert \mathcal{F}_t\right]  \frac{\mathbbm{1}_{\{\tau>T\}}}{Q\left(\tau>t \mid \mathcal{F}_t\right)}.
\end{equation*}
Since $\exp \left(-\int_t^T r_s d s\right)$ is $\mathcal{F}_t$-conditionally independent of $\mathbbm{1}_{\{\tau>T\}}$ (hypothesis \eqref{hyp:H3}) and using again $\mathbb{E}^{\mathbb{Q}}\left[\exp \left(-\int_t^T r_s d s\right) \bigg\vert \mathcal{F}_t\right]=P(t, T)$, we write: 
\begin{equation*}
    \begin{aligned}
    H_2 & =P(t, T)  \mathbbm{1}_{\{\tau>T\}} \frac{\mathbb{E}^{\mathbb{Q}}\left[\mathbbm{1}_{\{\tau>T\}} \mid \mathcal{F}_t\right]}{Q\left(\tau>t \mid \mathcal{F}_t\right)} \\
    & =P(t, T)  \mathbbm{1}_{\{\tau>T\}}  \frac{Q\left(\tau>T \mid \mathcal{F}_t\right)}{Q\left(\tau>t \mid \mathcal{F}_t\right)}.
    \end{aligned}
\end{equation*}
For the first time in our process, we now use the fact that we are in the context of intensity models $Q(\tau>t)$ is the survival probability and we know that $Q\left(\tau>T \mid \mathcal{F}_t\right)=\mathbb{E}^{\mathbb{Q}}\left[\exp \left(-\int_0^T \lambda_s d s\right) \vert \mathcal{F}_t\right]$. It is well known that:
\begin{equation*}
\frac{Q\left(\tau>T \mid \mathcal{F}_t\right)}{Q\left(\tau>t \mid \mathcal{F}_t\right)}=\mathbb{E}^{\mathbb{Q}}\left[\exp \left(-\int_t^T \lambda_s d s\right) \bigg \vert \mathcal{F}_t\right],
\end{equation*}
which implies that
\begin{equation}
    \label{eq:risky3}
    H_2=P(t, T) \mathbbm{1}_{\{\tau>T\}} \mathbb{E}^{\mathbb{Q}}\left[\exp \left(-\int_t^T \lambda_s d s\right) \bigg \vert \mathcal{F}_t\right].
\end{equation}
For instance, this is equivalent to result (3.3) of \citep{jeanblanccir}. Assuming that the default does not occur in within our observation window we take $\mathbbm{1}_{\tau>T} = 1$, we find the expression of the risky bond price in a intensity modelisation context.
We draw the reader's attention to the fact that, so far, we have not utilized the CIR++ model.
\begin{equation}
\label{eq:HtPt}
    H(t, T)=P(t, T)\left\{\delta+(1-\delta) \mathbb{E}^{\mathbb{Q}}\left[\exp \left(-\int_t^T \lambda_s d s\right) \bigg \vert \mathcal{F}_t\right]\right\}.
\end{equation}
The only component that remains to be determined is $\mathbb{E}^{\mathbb{Q}}\left[\exp \left(-\int_t^T \lambda_s ds\right) \Bigg\vert \mathcal{F}_t\right]$. Its expression can be directly obtained through the identification we have already made between the CIR++ intensity model and the CIR++ model for short rates. Since the survival probability (described in our model by the expectation above) is equivalent to the zero-coupon bond price, they share the same expression. We recall that for interest rate modeling under the CIR++, we have the following dynamic for the short-term interest rates:
\begin{equation*}
    \begin{cases}
        & r(t) = x(t) + \varphi(t)  \\
        & dx(t) = \kappa\left(\theta-x_t\right)dt + \sigma\sqrt{x_t}dW_t.
    \end{cases} 
\end{equation*}
The expression of the zero-coupon bond price in a CIR++ is very common in the literature (see for example Chapter 3.9 of \citep{brigo2006interest}) and is given by:
\begin{equation*}
    \mathbb{E}^{\mathbb{Q}}\left[\exp \left(-\int_t^T r_s d s\right) \bigg \vert \mathcal{F}_t\right] = P(t,T)=\frac{P^m(0,T)}{P^m(0,t)}\frac{A(0,t)}{A(0,T)}\frac{e^{-B(0,t)x_0}}{e^{-B(0,T)x_0}}A(t,T)e^{-B(t,T)[r(t)-\varphi(t)]}.
\end{equation*}
In the above, $P^m(0,t)$ is the market quote of the initial zero-coupon bond price maturing at $t$ and we have,
\begin{equation}
    \label{eq:AandB_CSPRN}
    \begin{aligned}
        A(t,T) &= \left( \frac{2he^{\frac{1}{2}(\kappa+h)(T-t)}}{2h+(\kappa+h)\left( e^{(T-t)h} - 1 \right )} \right)^{\frac{2\kappa\theta}{\sigma^2}};\\
        B(t,T) &= \frac{2\left( e^{(T-t)h} - 1 \right )}{2h+(\kappa+h)\left( e^{(T-t)h} - 1 \right )};\\
        h &= \sqrt{\kappa^2+2\sigma^2}.
    \end{aligned}
\end{equation}
Analogously, for our survival probability $S(t,T)$ we write :
\begin{equation}
    \label{eq:surv}
    \mathbb{E}^{\mathbb{Q}}\left[\exp \left(-\int_t^T \lambda_s d s\right) \bigg \vert \mathcal{F}_t\right] = S(t,T)=\frac{S^m(0,T)}{S^m(0,t)}\frac{A(0,t)}{A(0,T)}\frac{e^{-B(0,t)y_0}}{e^{-B(0,T)y_0}}A(t,T)e^{-B(t,T)[\lambda(t)-\psi(t)]}.
\end{equation} 
where $S^m(0,\cdot)$ is the market's implied\footnote{Unlike $P^m(0,t)$, survival probabilities are not usually directly quoted in the market. They must be inferred from relevant instruments such as CDS. This is discussed in more detail in the calibration section (\ref{sec:calib}).} survival probability at $t_0$.

The last step is to find the expression for $\psi(t)$. We have already stated that $\psi(t)$ should allow us to replicate the market's survival probabilities, $S^m(t)$. To achieve this, the model must satisfy $S^m(t) = S(0,t), \forall t \in [0,T]$. Here, we denote the market probability as $S^m(t)$, a uni-dimensional function, instead of $S^m(0,t)$ to emphasize that we are set at the initial time and $S^m$ is therefore a market observation. We can then write:
\begin{equation*}
	\begin{aligned}
		S^m(t)=S(0,t) &=\mathbb{E}^{\mathbb{Q}}\left[ e^{-\int_0^t{\lambda_s ds}} \right]\\
		&=\mathbb{E}^{\mathbb{Q}}\left[ e^{-\int_0^t{y_s + \psi(s) ds}}  \right]\\
& = e^{-\int_0^t{\psi(s) ds}}\mathbb{E}^{\mathbb{Q}}\left[ e^{-\int_0^t{y_s ds}} \right].
	\end{aligned}
\end{equation*}
$\mathbb{E}^{\mathbb{Q}}\left[ e^{-\int_0^t{y_s ds}} \right]$ is the price of a zero-coupon bond under a CIR model (or of the survival probability under a CIR intensity model). The literature (see the original CIR paper, \citep{cox1985theory} for instance) gives $\mathbb{E}^{\mathbb{Q}}\left[ e^{-\int_0^t{y_s ds}} \right] = A(0,t)e^{-B(0,t)y_0}$, so we can write:  
\begin{equation*}
	\begin{aligned}
		S^m(t) &= e^{-\int_0^t{\psi(s) ds}}A(0,t)e^{-B(0,t)y_0}.
	\end{aligned}
\end{equation*}
From our definition of the default intensity, it is natural to write the survival probability as a function of the cumulative hazard rate $S^m(t) = e^{-\Lambda^m(t)}$, where analogously to equation \eqref{eq:cumhr_CSPRN}, $\Lambda^m(t)  = \int_0^t{\lambda^m(u) du}$, where $\lambda^m(t)$ denotes the initial market default intensity for the horizon $t$. Thus, we get:
\begin{equation*}
	\begin{aligned}
	\int_0^t{\psi(s) ds} &= \Lambda^m(t)+\operatorname{ln}(A(0,t))-B(0,t)y_0, 
	\end{aligned}
\end{equation*}
and consequently, 
\begin{equation*}
	\begin{aligned}
 \psi(t) &= \lambda^m(t)+D(t)-y_0E(t)
	\end{aligned}
\end{equation*}
where 
\begin{equation*}
	\begin{aligned}
 D(t) &:= \frac{d}{dt}ln(A(0,t)) = \frac{2\kappa\theta}{\sigma^2}\left[ \frac{1}{2}(\kappa+h)-\frac{h(\kappa+h)e^{th}}{2h + (\kappa+h)(e^{th}-1)} \right]\\
 E(t) &:= \frac{d}{dt}B(0,t) = \frac{4h^2 e^{th}}{\left[ 2h+(\kappa+h)(e^{th}-1) \right]^2}.
	\end{aligned}
\end{equation*}
Finally, combining these results with equations \eqref{eq:HtPt} and \eqref{eq:surv} we have our equation for the price of defaultable bonds:
\begin{equation*}
    \begin{dcases}
        H(t, T) = P(t,T) \left[\delta+(1-\delta) \frac{S^m(0,T)}{S^m(0,t)}\frac{A(0,t)}{A(0,T)}\frac{e^{-B(0,t)y_0}}{e^{-B(0,T)y_0}}A(t,T)e^{-B(t,T)[\lambda(t)-\psi(t)]} \right]\\
        \psi(t) = \lambda^m(t)+D(t)-y_0E(t). 
    \end{dcases}
\end{equation*}
\end{proof}

\paragraph{Comment on the form of the defaultable bond price equation} \mbox{}\\
From equation \ref{eq:HtPt} in the proof of Theorem \ref{theo:ExpDefBond} we have that:
\begin{equation*}
    H(t, T)=P(t, T)\left[\delta+(1-\delta) S(t, T)   \right].
\end{equation*}
This equation allows for an intuitive analysis of the expression. It guarantees that for any given recovery rate, the closer the survival probability is to $1$, the closer the risky bond price $H(t,T)$ is to the risk-free bond price $P(t,T)$. In particular, if the survival probability is $1$, meaning that the risky bond is virtually default-free, then the two bonds have the same price. Additionally, this equation implicitly highlights the risk premium that investors demand for holding a defaultable bond over a risk-free bond, as it ensures that $H(t,T)$ is always lower than $P(t,T)$. To exploit this equation to propagate the term structure of a defaultable bond price, one needs to combine equation \eqref{eq:defbondfinal} with a model for the pricing of the risk-free bond. Specifically, one needs to define the dynamics of the short-term interest rates. For instance, if one assumes the instantaneous rate to be modeled by a Hull-White model, then it is well known \citep{brigo2006interest} that we would have:
$$
dr(t) = [\vartheta(t) - a^rr(t)]dt + \sigma^r dW^r(t),
$$
where \(a^r\) and \(\sigma^r\) are positive constants and \(\vartheta\) is chosen so as to exactly fit the term structure of interest rates being currently observed in the market and $W^r(t)$ is a $\mathbb{Q}-Brownian$ motion. The bond price $P(t,T)$ would be:

$$
P(t,T) = A^{HW}(t,T)e^{-B^{HW}(t,T)r(t)},
$$
where 
$$
\begin{aligned}
B^{HW}(t,T) &= \frac{1}{a} \left[ 1 - e^{-a(T-t)} \right], \\
A^{HW}(t,T) &= \frac{P^m(0,T)}{P^m(0,t)} \exp \left[ B^{HW}(t,T) f^m(0,t) - \frac{\sigma^2}{4a} \left( 1 - e^{-2at} \right) B^{HW}(t,T)^2 \right],
\end{aligned}
$$
and $P^m(0,t)$ is the initial market bond price at time 0.

\section{Modeling the Term Structure of Credit Spreads}
\label{sec:RNPmodel}
Let us now focus on the main aim of this research paper: the term structure of credit spreads. We recall that our modeling framework represents default intensities using a CIR++ model and then finds the equation that defines credit spreads in that context. We maintain the same setting and definitions as previously established. As we have already stated, the price of the risk-free zero-coupon bond is:
\begin{equation*}
    \label{eq:riskfreebond2}
    P(t, T)=\mathbb{E}^{\mathbb{Q}}\left[\exp \left(-\int_t^T r_s d s\right) \bigg\vert \mathcal{F}_t\right].
\end{equation*}

We recall that in our setup we consider, $(r_t)_{0\le t\le T}$ to be a $(\mathcal{F}_t)_{0\le t\le T}$-measurable process, which is $\mathcal{F}_t$-conditionally independent of the default intensity and therefore any $\tau$-defined process. The fair-price of a risky-bond with recovery rate $\delta\in(0,1)$ is defined by
\begin{equation*}
    H(t, T)=\mathbb{E}^{\mathbb{Q}}\left[\exp \left(-\int_t^T r_s d s\right)\left(\delta \mathbb{1}_{\{\tau \leq T\}}+\mathbb{1}_{\{\tau>T\}}\right) \bigg\vert \mathcal{G}_t\right],
\end{equation*} 
and given explicitly by Theorem \ref{theo:ExpDefBond}. 

The credit spread represents the difference between the yield of a risky asset, denoted as $ Z_r(t,T) $, and the yield of a risk-free asset, denoted as $ Z(t,T) $. Typically, the risk-free asset can be government bonds, which are considered to have minimal default risk. In contrast, the risky asset could be a corporate bond or a bond issued by a financial institution such as a bank, which inherently carries a higher default risk. The relationship between the yields and the prices of these assets at time $t$ for a bond maturing at time $T$ is given by the following expressions:
$$
P(t,T) = e^{-(T-t)Z(t,T)}; \; H(t,T) = e^{-(T-t)Z_r(t,T)}.
$$
The exponential function in these classical formulas reflects the continuous compounding of interest rates. The credit spread, $Z_r(t,T) - Z(t,T)$, thus quantifies the additional yield that investors demand for taking on the additional risk associated with the corporate or bank-issued bond compared to a government bond. This spread is a crucial metric in credit risk analysis, reflecting the market's perception of the default risk and the overall creditworthiness of the issuer. Therefore,  the credit spread $\operatorname{Sp(t,T)}$ of maturity $T$ and evaluated at time $t$  is given by
\begin{equation}
    \label{eq:spreaddef}
        Sp(t, T) := Z(t,T)-Z_r(t,T) = -\frac{1}{T-t}\operatorname{log}\left[\frac{H(t,T)}{P(t,T)}\right].
\end{equation}

\begin{theorem}
    \label{theo:ExpSp}
    Under the previous notations, the value of the credit spread of maturity $T$ at time $t$, $Sp(t,T)$ is given by: 
    \begin{equation}
    \label{eq:spfinal}
    \boxed{
        \begin{dcases}
            \operatorname{Sp}(t, T) =-\frac{1}{T-t} \ln \left[\delta+(1-\delta) \frac{S^m(0,T)}{S^m(0,t)}\frac{A(0,t)}{A(0,T)}\frac{e^{-B(0,t)y_0}}{e^{-B(0,T)y_0}}A(t,T)e^{-B(t,T)[\lambda(t)-\psi(t)]} \right]\\
            \psi(t) = \lambda^m(t)+D(t)-y_0E(t), 
        \end{dcases}
    }
    \end{equation}
    where, 
    \begin{equation*}
    \begin{aligned}
        A(t,T) &= \left( \frac{2he^{\frac{1}{2}(\kappa+h)(T-t)}}{2h+(\kappa+h)\left( e^{(T-t)h} - 1 \right )} \right)^{\frac{2\kappa\theta}{\sigma^2}}\text {, }
        B(t,T) = \frac{2\left( e^{(T-t)h} - 1 \right )}{2h+(\kappa+h)\left( e^{(T-t)h} - 1 \right )}\\
        h &= \sqrt{\kappa^2+2\sigma^2}\\
        D(t) &= \frac{d}{dt}ln(A(0,t))\text{, }E(t) = \frac{d}{dt}B(0,t).
    \end{aligned}
\end{equation*}
\end{theorem}
$S^m(0,t)$, $A(t,t)$, $B(t,T)$, $D(t,T)$, $E(t,T)$, $\lambda^m(t,T)$ and $S^m(0,t)$ are all specified in Theorem \ref{theo:ExpDefBond}.

\begin{proof}
Once the proof of Theorem \ref{theo:ExpDefBond} is given, the proof of Theorem \ref{theo:ExpSp} is quite straightforward. Indeed, on one hand, since
\begin{equation*}
    H(t, T)=P(t, T)\left\{\delta+(1-\delta) \mathbb{E}^{\mathbb{Q}}\left[\exp \left(-\int_t^T \lambda_s d s\right) \bigg \vert \mathcal{F}_t\right]\right\},
\end{equation*}

we naturally have
\begin{equation*}
    \label{eq:sp}
        \operatorname{Sp}(t, T)=-\frac{1}{T-t} \ln \left[\delta+(1-\delta) \mathbb{E}^{\mathbb{Q}}\left[\exp \left(-\int_t^T \lambda_s d s\right) \bigg \vert \mathcal{F}_t\right]\right].
\end{equation*}

On the other hand, we had already justified for Theorem \ref{theo:ExpDefBond} that the survival probability analytical expression in our setup was:
\begin{equation*}
    \mathbb{E}^{\mathbb{Q}}\left[\exp \left(-\int_t^T \lambda_s d s\right) \bigg \vert \mathcal{F}_t\right] = S(t,T)=\frac{S^m(0,T)}{S^m(0,t)}\frac{A(0,t)}{A(0,T)}\frac{e^{-B(0,t)y_0}}{e^{-B(0,T)y_0}}A(t,T)e^{-B(t,T)[\lambda(t)-\psi(t)]}.
\end{equation*} 

The above leads to our equation for the diffusion of the term structure of credit spreads:
\begin{equation*}
    \begin{dcases}
        \operatorname{Sp}(t, T) =-\frac{1}{T-t} \ln \left[\delta+(1-\delta) \frac{S^m(0,T)}{S^m(0,t)}\frac{A(0,t)}{A(0,T)}\frac{e^{-B(0,t)y_0}}{e^{-B(0,T)y_0}}A(t,T)e^{-B(t,T)[\lambda(t)-\psi(t)]} \right]\\
        \psi(t) = \lambda^m(t)+D(t)-y_0E(t). 
    \end{dcases}
\end{equation*}
\end{proof}

\section{Model Calibration}
\label{sec:calib}
Once our credit spread model is well established, it seems natural to develop the associated calibration approach. A model can only be as good as its parameters are well chosen to be representative of a specific context. 
Many calibration approaches exist in the literature (see e.g. \cite{}). An efficient approach that has already been tested on CIR intensity model can be the deep calibration developed in \citep{dc}. We will introduce another calibration approach as suitable  as the deep calibration but with a focus on a risk management perspective. 

\paragraph{Calibration Process}\mbox{}\\
A natural use of this model will be asset fair-value calculations, either for pricing issues or for market risk management. In the first case, the intuitive calibration might be to fit the last few credit spreads values or to fit market prices of some assets. This idea would be consistent with most calibration processes for interest rates models, where an indeed common practice is to fit the prices of some interest rate derivatives (see for example \citep{hull2001general} or more recently \citep{russo2019calibration}). However for a more risk assessment oriented calibration, finding the best parameters to only fit a set of given values may not guarantee the ability to cover all the historical variations of a credit spreads. Indeed, in this latter case, a conservative practice is to make sure the model is representative of the widest variations of credit spreads observed in the market. These variations are related to the historical volatility of credit spreads. We therefore choose a tractable approach that consists in calibrating on the historical volatility of the market default intensity curve. This choice is all the more relevant as the literature has already established the fact that in a CIR++ intensity configuration, the default intensity can be seen as an instantaneous credit spread (see for example Chapter 21.1 of \cite{brigo2006interest}). The model squared volatility (which is the variance) of the default intensity is given by: 
\begin{equation}
    \label{eq:varlambda}
    \operatorname{Var}_{\lambda}(T; \kappa,\theta, \sigma, y_0)=y_0 \frac{\sigma^2}{\kappa} (e^{- \kappa T}-e^{-2\kappa T}) + \frac{\theta \sigma^2}{2 \kappa}(1-e^{- \kappa T})^2.
\end{equation}
We denote by $\operatorname{Vol}_{\lambda}^m(T)$ the market volatility of the default intensity. We will try to minimize the error between the theoretical and historical volatilities, and we set $\operatorname{Vol}_{\lambda}(T;\dots)=\sqrt{\operatorname{Var}_{\lambda}(T;\dots)}$. Regarding the choice of the objective function, it turns out that choosing the Sum of Squared Relative Error (SSRE)\footnote{This error measurement metric is not common, however it is the sum and root-less version of the root mean squared relative error found in \citep{errors}.} achieves the best results as it leads to good and consistent calibrations (see figures \ref{fig:calibbase} and \ref{fig:calibstress}). Let $\Theta=\{\kappa, \theta, \sigma, y_0\}$ be
the set of parameters to calibrate and $\Theta^*=\{\kappa^*, \theta^*, \sigma^*, y_0^*\}$ the set of calibrated parameters. Let $M$ be the number of maturities used for the calibration. We calibrate on the volatilities for each of the $M$ maturities $\{T_i\}_{i\in \llbracket1, M\rrbracket}$. Therefore, we solve: 
\begin{equation}
    \label{eq:calib}
    \Theta^{*} = \underset{\Theta}{\arg\min}\sum_{i=1}^M{\left[\frac{\operatorname{Vol}_{\lambda}^m(T_i) - \operatorname{Vol}_{\lambda}(T_i;\,\Theta)}{\operatorname{Vol}_{\lambda}^m(T_i)}\right]^2}.
\end{equation}
We have included $y_0$ in the set of parameters to calibrate. This is not mandatory. The only relation we need to ensure is that whatever the choice of $y_0$, and the way it is obtained, the equality $\psi(0) = \lambda^m(0)-y_0$ must hold. Adding it to the parameters to calibrate gives us another degree of freedom to achieve the best calibration possible. From equation \eqref{eq:varlambda} we get the theoretical volatility. To obtain the historical volatilities of default intensity, we apply the procedure described in \citep{volat} for interest rates to each default intensity horizon. In that article, the authors use weekly data. For each week, the volatility calculation is made using that week and the previous 12 or 51 weeks and the volatility is represented by the standard deviation over the period for each week. We chose 51 weeks as we found it to give more stable results. This provides us with a time series of volatilities of default intensity for each horizon. Finally we need to choose one value driven by this time series to be representative of the whole volatility. Many solutions seem reasonable, for instance taking the average value. In this paper, we make the more conservative choice—again, to adopt a more risk-oriented perspective—by taking the highest values. That is to say that the maximal weekly volatility obtained via the process described above will represent the historical volatility over that period. This choice is very conservative but it is the one that guarantees that any economic stress encountered within the calibration process will be taken into account. We recommend practitioners who want to be less conservative to take the median or the average value of the window. Also, as we have already stated when this model is used for pricing issues, calibrating on the last few values of a given quantity (prices for instance) would be more suited.

\paragraph{Data set}\mbox{}\\
For the sake of visualization, calibration, and diffusion, we will focus on the credit spread of the French bank, Credit Agricole. We begin by observing Credit Agricole's risky bond yield and the risk-free one, which we take as the French government bond. Both data sets can be obtained via private market data providers (Bloomberg, Refinitiv Eikon, \dots). The data is daily and covers the period from January $01^{\operatorname{st}}$, 2009, to January $01^{\operatorname{st}}$, 2024. The values are in basis points (BPs).
\begin{figure}[H]
    \centering
    \includegraphics[width=15.5cm, height=6.cm]{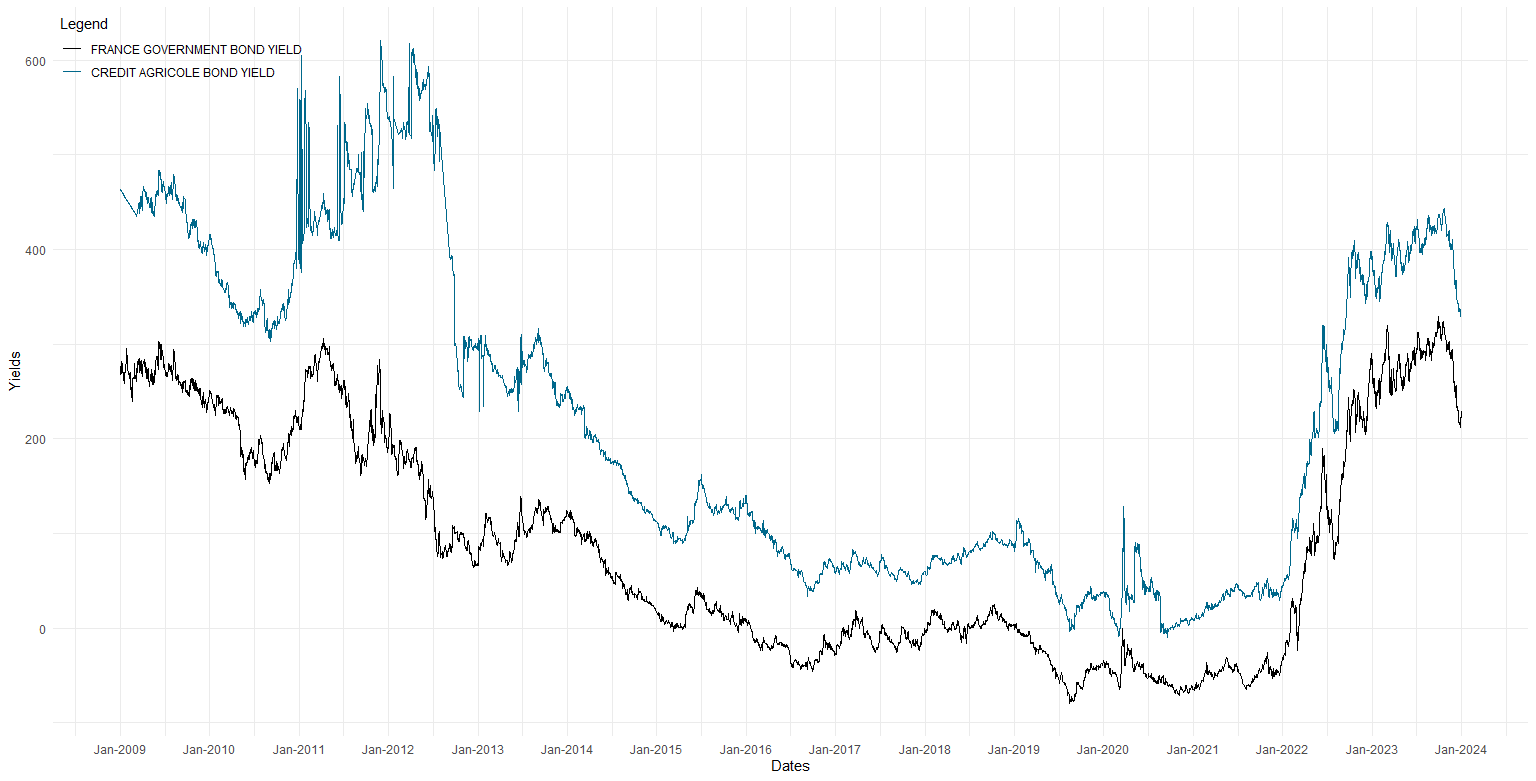}
    \vspace{-0.4cm}
    \caption{In BP, risk-free 5Y bond yield in blue (french government bond) and risky 5Y bond yield (Credit Agricole)}
    \label{fig:yields}
\end{figure}
Then, we observe the 5Y credit spread. It is computed as the difference between the French government bond yield and Credit Agricole's bond yield. Note that as of December 2023, S\&P rates the French government as AA and Credit Agricole as AA-, as of September 2023. 

\begin{figure}[H]
    \centering
    \includegraphics[width=15.5cm, height=6.cm]{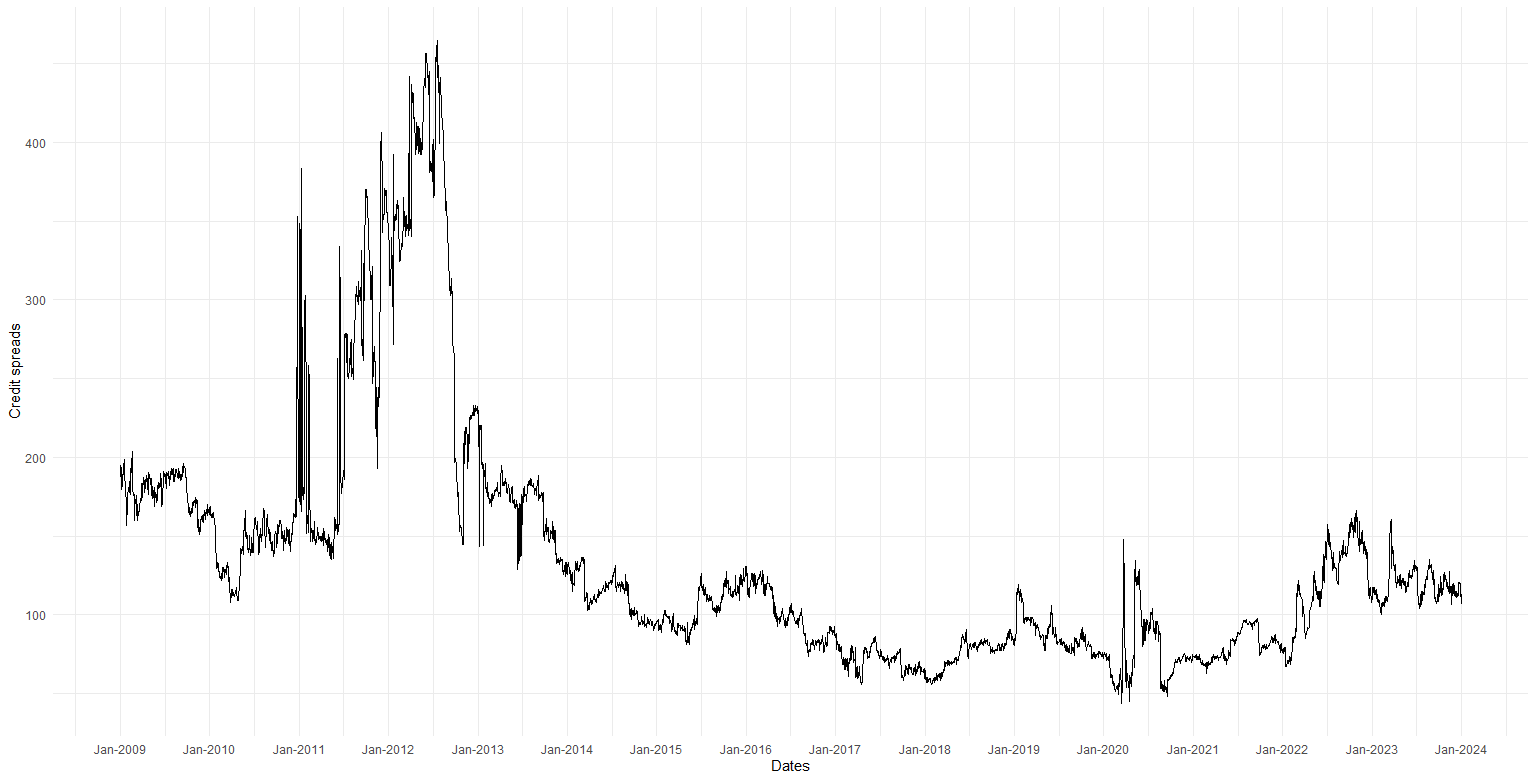}
    \vspace{-0.4cm}
    \caption{In BP, Credit Agricole's 5Y credit spread}
    \label{fig:spread}
\end{figure}
Since for our model we need both default intensities and survival probabilities for the diffusion of credit spreads via equation \eqref{eq:spfinal} and for the calibration, we can use the CDS price curve of Credit Agricole to bootstrap said default intensities and survival probabilities. The CDS curve has the same range as the yields and credit spreads (2009-01-01 to 2024-01-01) and is provided in Appendix \ref{sec:appendixcalibdata}. The bootstrapping procedure is not discussed in this paper but is explained in Chapter 22.3 of \citep{brigo2006interest}. Practically, in this paper, we have used R's library \textit{Credule}\footnote{\url{https://CRAN.R-project.org/package=credule}}. We also present in Section \ref{sec:mcrnp} a way to imply the survival probabilities from credit spreads. Figure \ref{fig:cds} below shows the CDS curve. Figures \ref{fig:di} and \ref{fig:srp}, provided in Appendix \ref{sec:appendixcalibdata}, show, respectively, the bootstrapped 5Y horizon default intensities and survival probabilities. Finally, we take the recovery rate to be $\delta = 40\%$.

\paragraph{Calibration Results}\mbox{}\\
We use the data set described in the previous paragraphs to compute, for each maturity, the annual volatility of the credit spreads. We use the set of maturities \{1Y, 3Y, 5Y, 7Y, 10Y\}, which are reputed to be more reliable (regarding the availability of quotations and liquidity issues). As we have at our disposal a large data set, and since, as we have already stated, in this paper—especially with the RW modeling—we are more risk management-oriented, we make two calibrations.

\begin{itemize}
    \item \textit{Global scenario}\\
    We aim for the model to be representative of the historical market situation. We take the last fifteen years as a reference to capture this historical behavior. Therefore, we calibrate using default intensities from 2009-01-01 to 2024-01-01. This data set ensures that we have sufficient data to represent the overall market behavior. This window includes stressed periods, such as the European/Greek Government Debt Crisis (see, for instance, \citep{eurozonecrisis}), as well as relatively calm periods.
    
    \item \textit{Present scenario}\\ ThisThis scenario is intended to represent the current behavior of credit spreads. We focus on the last two years, from 2022-01-01 to 2024-01-01. This calibration does not include any visually extreme variations in credit spreads, even though it captures part of the 2021-2022 post-pandemic inflation period (see, for example, \cite{inflation}). As shown in Figure \ref{fig:yields}, the risk-free and risky yields seem to have reacted similarly to the crisis, resulting in relatively stable credit spreads, as depicted in Figure \ref{fig:spread}. This scenario is therefore expected to yield more favorable parameters. 
\end{itemize}

The SSRE errors, as determined by equation \eqref{eq:calib}, are $1.009 \times 10^{-3}$ and $3.276 \times 10^{-3}$, respectively, for the \textit{global scenario} and the \textit{present scenario}. The parameters are given, respectively, by Tables \ref{table:resbase} and \ref{table:resstress}. Figures \ref{fig:calibbase} and \ref{fig:calibstress} show the calibration's fittings.
\begin{table}[H]
    \centering 
    \begin{tabular}{c c c c}
        \hline \hline 
         & & & \\ [-0.4em] 
        $\kappa$ & $\theta$ & $\sigma$ & $y_0$ \\ [0.4em] 
        \hline \hline 
        & & & \\
        $5.138 \times 10^{-1}$ & $1.497 \times 10^{-2}$ & $8.904 \times 10^{-2}$ & $4.348 \times 10^{-2}$ \\ [0.5em] 
        & & & \\
    \end{tabular}
    \caption{Parameters calibrated for the \textit{Global scenario}}
    \label{table:resbase}
\end{table}
\begin{table}[H]
    \centering 
    \begin{tabular}{c c c c}
        \hline \hline 
         & & & \\ [-0.4em] 
        $\kappa$ & $\theta$ & $\sigma$ & $y_0$ \\ [0.4em] 
        \hline \hline 
        & & & \\
        $9.186 \times 10^{-2}$ & $5.519 \times 10^{-4}$ & $1.006 \times 10^{-2}$ & $3.074 \times 10^{-2}$ \\ [0.5em] 
        & & & \\
    \end{tabular}
    \caption{Parameters calibrated for the \textit{Present scenario}. Volatilities are expressed in BPs.}
    \label{table:resstress}
\end{table} 

\vspace{-9.5mm}
\begin{figure}[H]
    \centering
    \begin{minipage}{0.5\textwidth}
        \begin{figure}[H]
            \centering
            \includegraphics[width=7.cm, height=6.5cm]{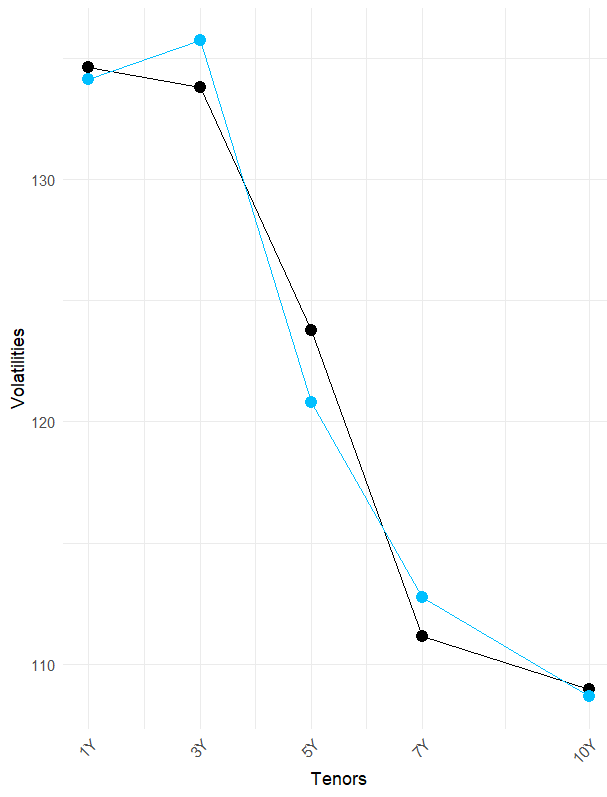} 
            \vspace{-1.5mm}
            \caption{Calibration's fitting for the \textit{global scenario}}
            \label{fig:calibbase}
        \end{figure}
    \end{minipage}\hfill
    \begin{minipage}{0.5\textwidth}
        \begin{figure}[H]
            \centering
            \includegraphics[width=8.25cm, height=6.5cm]{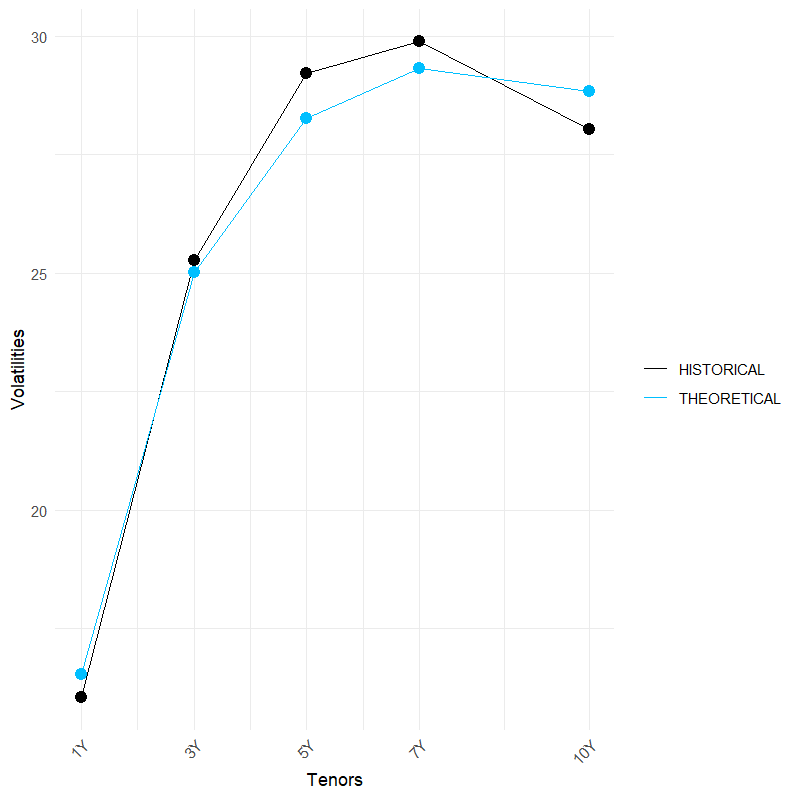} 
            \vspace{-5.5mm}
            \caption{Calibration's fitting for the \textit{present scenario}. Volatilities are expressed in BPs.}
            \label{fig:calibstress}
        \end{figure}
    \end{minipage}
\end{figure} 
It is interesting to compare the parameters in the two scenarios. In fact, we observe a very appreciable behavior: both the volatility, $\sigma$, and the convergence parameter, $\theta$, are considerably decreased for the present (stress-less) scenario. Practitioners can use any time range, expert-wisely chosen, to calibrate their parameters, whether it is for the base scenario, the present one, or even a stressed one. This is a considerable advantage of calibrating using quantities that are representative of the behavior over an entire period. Indeed, the significant increase in volatility for the global scenario, where volatilities range from 110 to 130 BPs as opposed to 20 to 30 BPs during the 2022-2024 period, is reflected in the calibration by the decrease of the parameters, especially $\sigma$. In the next sections, we will only use the global scenario parameters.

\paragraph{Comment on the shape of the volatilities} \mbox{}\\
It is important to make a quick comment on the shapes of volatilities of default intensity. We observe a decreasing curve with respect to maturities for the global scenario and an increasing one for the present scenario. We have conducted many tests on different time frames, and the observation is always the same: if the time frame's main event is a stress period (meaning that the time frame is short and includes a minor stress or the time frame is large and includes a sufficiently large stress, like our global scenario), the volatility decreases with maturity; otherwise, it increases. As far as we know, this behavior is not studied in the literature, likely because it doesn't seem to have a significant impact. In this paper, we do not try to find the rationale behind this observed behavior. However, we are inclined to guess a link with the inversion of yield curves during stress periods (see, for instance, \citep{invertedyieldcurve}), as we have also observed an inversion of the credit spread curve under stress periods. 

\section{Diffusion of credit spreads in the model}
\label{sec:mcrnp}
After having defined, then calibrated the model, we can now diffuse the credit spreads. Their diffusion is fairly straightforward using equation \eqref{eq:spfinal}. The first step is to diffuse the default intensity $\lambda(t)$ and then apply equation \eqref{eq:spfinal}. For the diffusion of $\lambda(t)$, we use the actual distribution of the solution of the CIR equation, $y(t)$, then we add the function $\phi(t)$. 
The probability density of $y(t)$ is given in \citep{cox2005theory} and is:
\begin{equation*}
    f(y(t), t \; ; \; y(s), s)=c e^{-w-v}\left(\frac{v}{w}\right)^{q / 2} I_q\left(2(w v)^{1 / 2}\right), 
\end{equation*}
where,  
$
    c := \frac{2 \kappa}{\sigma^2\left(1-e^{-\kappa(t-s)}\right)}, w := c y(s) e^{-\kappa(t-s)}, v := c y(t), \text{ and } q := \frac{2 \kappa \theta}{\sigma^2}-1
$
and $I_q(\cdot)$ is the modified Bessel function of the first kind of order $q$. The distribution function is the noncentral chi-square, $\chi^2[2 c y(s) ; 2 q+2,2 w]$, with $2 q+2$ degrees of freedom and parameter of noncentrality $2 w$ proportional to the current spot rate. In practice, in this paper for the simulations, we used R’s library sde\footnote{\href{https://CRAN.R-project.org/package=sde}{https://CRAN.R-project.org/package=sde}}. Note that this distribution is conditional on past values. Therefore, even when we are not using a diffusion scheme, we need to set a time step and keep track of at least one previous value. In all the simulation in this paper, the time step is one week. We start from the configuration of the market at the simulation date, say January $1^{st}$, 2024. Using the process described above, we simulate 20,000 paths of credit spreads for $T= 2 (years)$ and with a time step of 1 week. Figure \ref{fig:avgdiff} shows the expectation (average) of credit spreads simulated at the origin (week 0) and at weeks 25, 50, 75, and 100. Figure \ref{fig:quantdiff} shows the 10\% and 90\% quantiles of the term structure of credit spreads simulated for weeks 25, 50, 75, and 100.

\vspace{-8.5mm}
\begin{figure}[H]
    \centering
    \begin{minipage}{0.495\textwidth}
        \begin{figure}[H]
            \centering
            \includegraphics[width=7.cm]{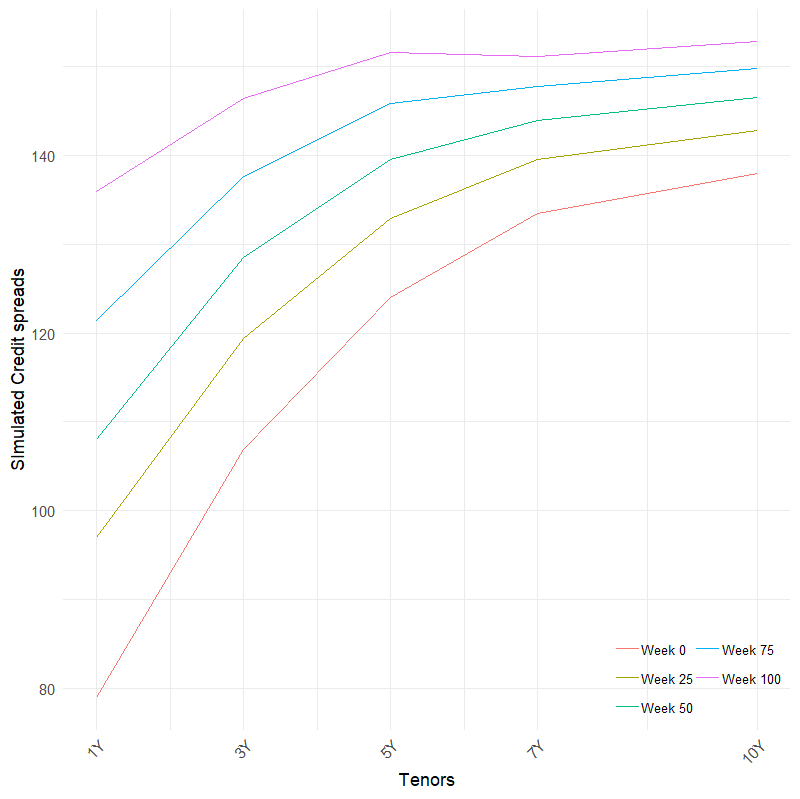} 
            \vspace{-1.5mm}
            \caption{\footnotesize In BP, average term structure of simulated credit spreads}
            \label{fig:avgdiff}
        \end{figure}
    \end{minipage}\hfill
    \begin{minipage}{0.495\textwidth}
        \begin{figure}[H]
            \centering
            \includegraphics[width=7.cm]{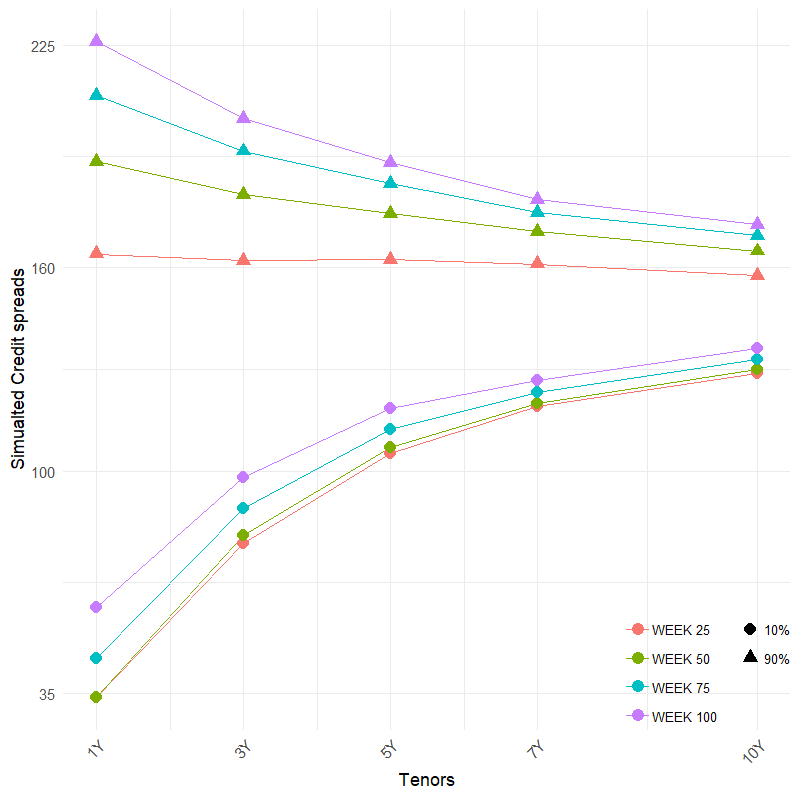} 
            \vspace{-1.7mm}
            \caption{\footnotesize In BP, quantiles 10\% and 90\% of simulated term structure of credit spreads}
            \label{fig:quantdiff}
        \end{figure}
    \end{minipage}
\end{figure} 
Figures \ref{fig:avgdiff} and \ref{fig:quantdiff} serve as the first back-tests of our model. They indeed show that our model produces realistic term structure of credit spreads over time and that the 10\%-90\% inter-quantile range seems coherent. Figure \ref{fig:histdiff5Y} shows the distribution of the 5Y credit spreads after about 1Y and 2Y (weeks 50 and 100). The figure also demonstrates a realistic shift in the credit spreads distribution over time. In subsection \ref{sec:backtest}, we conduct a more robust back-test of our approach.
\begin{figure}[H]
    \centering
    \includegraphics[width=15.5cm, height = 6.5cm]{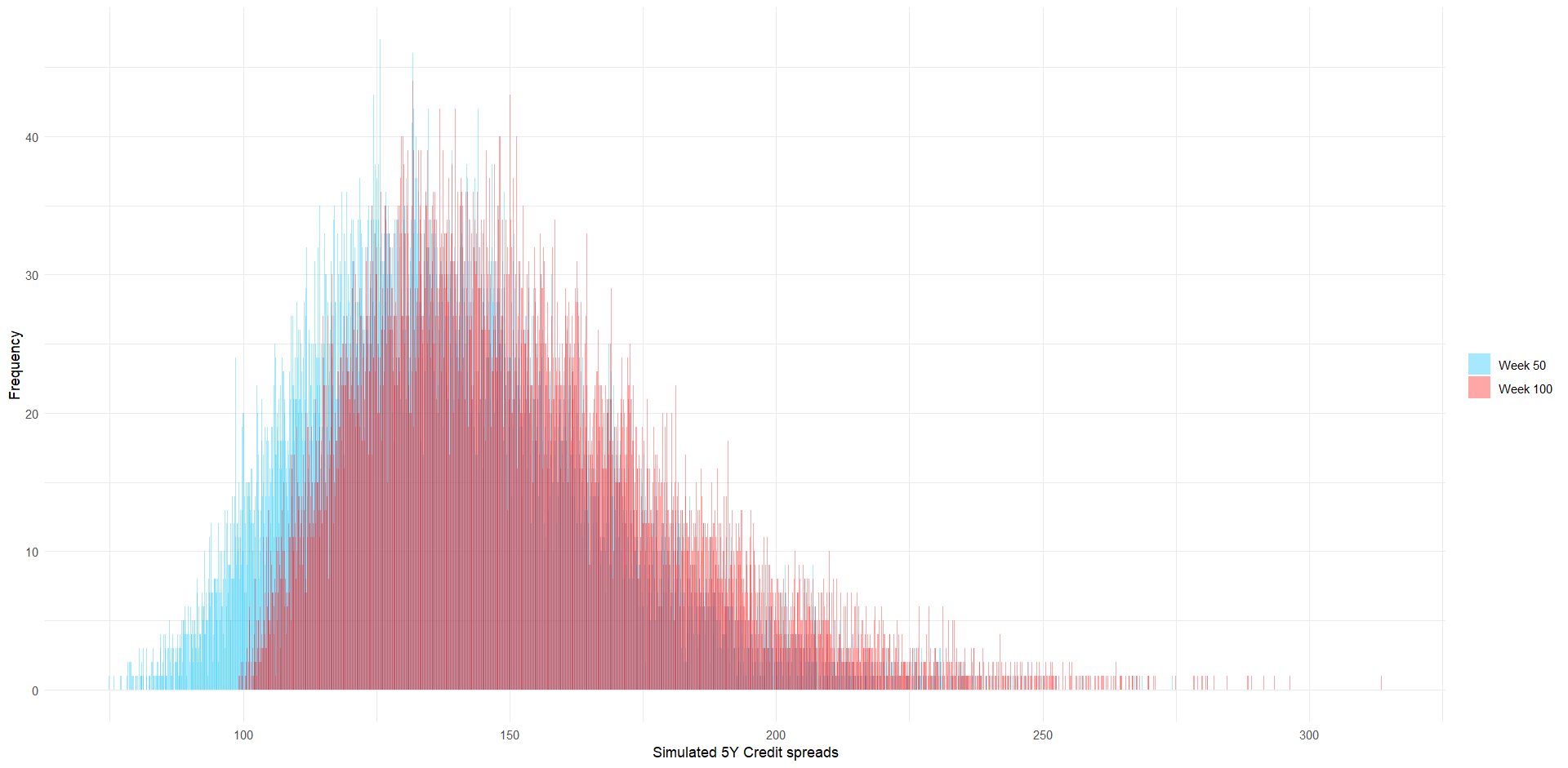}
    \caption{In BP, histogram of the simulated 5Y credit spreads after 50 weeks and 100 weeks.}
    \label{fig:histdiff5Y}
\end{figure}

\paragraph{Fitting the initial market term structure} \mbox{} \\ 
As previously discussed, the model fits the term structure of the cumulative hazard rate, which inherently corresponds to fitting the structure of survival probabilities. An advantageous consequence of this is the model's capability to fit the prices of Credit Default Swaps (CDS) for the risky asset, as survival probabilities are often implied from CDS values. Alternatively, in this model, one may opt to fit the initial credit spread values instead of CDS. In this scenario, the model permits the implied derivation of initial survival probabilities from the initial credit spreads. Equation \eqref{eq:spfinal} allows us to express the initial market survival probabilities, denoted as $S^m(0, T)$, as a function of the initial market credit spreads, $\operatorname{Sp}^m(0,T)$ and independently of the parameters of the model:
\begin{equation}
    \label{eq:fitsp0}
    S^m(0, T)=\frac{e^{-T \operatorname{Sp}^m(0, T)}-\delta}{1-\delta}
\end{equation}
It is noteworthy that these $S^m(0,T)$ values are always below 1 for positive credit spreads. However, ensuring their positivity requires a non-restrictive condition, which stipulates $\operatorname{Sp}^m(0,T) < -\operatorname{ln}(\delta)/T$. With the $\delta$ value of 40\% utilized in this article, this condition remains non-restrictive, allowing for credit spreads up to 900 basis points, even for maturities as long as 10 years.

\section{Back-testing the model}
\label{sec:backtest}
As we have already stated, the figures in the previous section show that the model indeed produces a realistic term structure of credit spreads over time. However, we want to introduce another measure of our model's coherence and consistency. Specifically, we aim to evaluate how our model's extreme values behave with respect to the actual values of credit spreads. This back-test will also serve as a validation of the calibration. The procedure is as follows:
\begin{enumerate}
    \item We have selected a suitable time horizon to analyze the actual observed credit spreads alongside the extreme values generated by our model. To ensure that the initial conditions adequately inform subsequent periods, we have capped the duration at 200 weekly intervals, roughly equivalent to a four-year time frame. Therefore, we take the period spanning from January 1, 2020, to January 1, 2024.
    \item We start from the market conditions at the beginning of the period (2020-01-01). Using the diffusion process presented in the previous section, we simulate 20,000 paths of credit spreads.
    \item For a specific maturity term, such as the 5-year tenor, we examine how various quantiles encompass the actual values of credit spreads.
\end{enumerate}

Figure \ref{fig:backtest} shows the 5Y credit spreads and the quantiles 1\%, 10\%, 20\%, 30\%, 70\%, 80\%, 90\% and 99\% of the simulated credit spreads. We can observe how at the 99\% level, only one historical credit spread value is beyond our simulation. This result demonstrates how conservative our model is and how much the credit spread values it gives are realistic. However it is important to note that this back-test algorithm is very restrictive and highly demanding. Credit spreads time series might have jumps or drops and especially important change of slopes that one might not be able to completely reproduce over time. We achieve the good results observed in figure \ref{fig:backtest} in large part thanks to the conservative calibration we have made in section \ref{sec:calib}, when choosing the calibration period as well as when choosing the representative volatility value. In fact we can actually see that the quantile 99\% being close to the peak between 2020 and 2021 implies that it is way above the values that come after. Also if the window analysis had included significant changes in slopes, this back-test could have missed them while the model would have remained sufficiently coherent in terms of shapes of credit spreads term structures or of range of values.

 \begin{figure}[H]
    \centering
    \includegraphics[width=15.75cm, height=7cm]{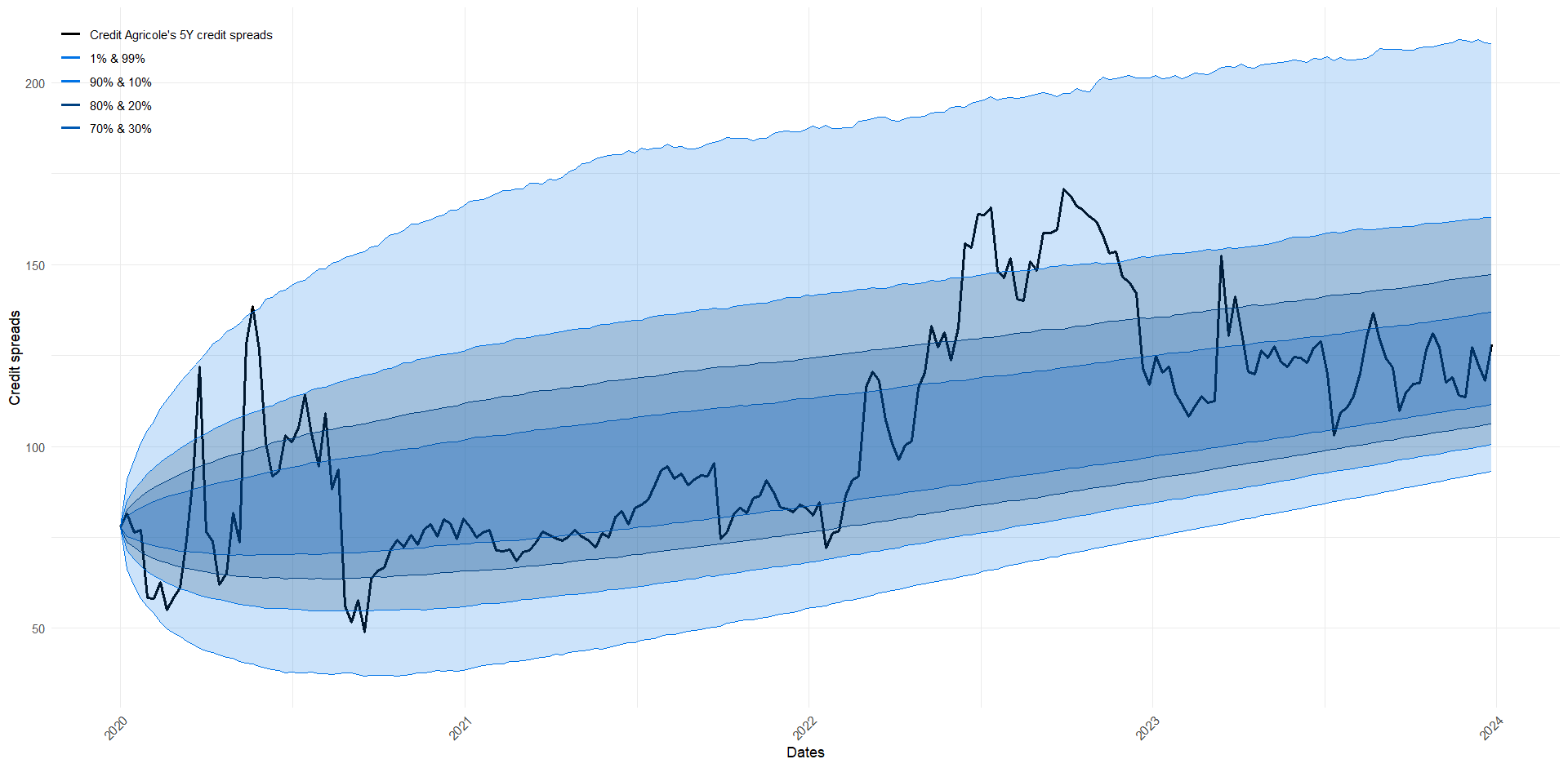}
    \caption{In BP, 5Y Credit Agricole's credit spreads observed values (black) and quantiles (shades of blue).}
    \label{fig:backtest}
\end{figure}
Note that for this back-test simulation, we have used equation \eqref{eq:fitsp0} to fit the initial term structure of credit spreads. As we can see, the initial value of the simulated 5Y credit spread matches the initial market value.

\section{Conclusion}

In this paper, we introduced a novel stochastic model for credit spreads using a CIR++ intensity model, addressing significant gaps in the current literature. Our model provides several key benefits that enhance its practical utility and theoretical robustness.

Firstly, the explicit expression of the credit spread under the CIR++ intensity model offers a powerful tool for practitioners. This explicit formulation simplifies the process of deriving credit spreads and allows for more straightforward applications in various financial computations and risk assessments. The ability to directly obtain the term structure of credit spreads is particularly valuable, as it enables a more accurate reflection of market conditions over different time horizons. This direct term structure derivation is crucial for pricing, risk management, and strategic financial planning, providing a comprehensive view of credit risk dynamics.

An important intermediary result of our research is the analytical expression for the price of a defaultable bond. This contribution not only enriches the understanding of defaultable bond pricing but also integrates seamlessly with our credit spread model. The simplicity and effectiveness of our calibration approach, which relies on the volatilities of default intensities, ensure that both models are robust and conservative. The calibration process was meticulously conducted by minimizing the Sum of Squared Relative Error (SSRE) between historical volatilities of default intensities by horizon and the theoretical expression of the standard deviation of the default intensity. This approach ensured the accuracy and reliability of the model, as shown by the back-test. The latter revealed how the extreme quantiles of the simulated credit spreads encompassed the actual historic credit spreads, further validating the model’s effectiveness.

In conclusion, the introduced stochastic model for credit spreads contributes to the literature by offering a robust, practical, and analytically sound framework. Future research could explore several exciting avenues to further enhance the model's applicability and depth. One potential direction is the application of the model to the pricing of options and the calculation of XVA (valuation adjustments), which are crucial for comprehensive risk management and derivative pricing. Additionally, incorporating more complexity into the model, such as introducing jumps or establishing correlations with other risk factors, could provide a more nuanced understanding of credit spreads and their interactions with broader market conditions.

Our work sets a solid foundation for more sophisticated credit spread modeling, promising significant advancements in financial risk management and economic forecasting. By continuing to refine and expand upon this model, researchers and practitioners alike can develop more effective strategies for navigating the complexities of credit risk.

\newpage 

\bibliographystyle{unsrtnat}
\bibliography{references}  

\newpage

\appendix

\section{Calibration Data}
\label{sec:appendixcalibdata}
Figure \ref{fig:cds} shows the historical values of Credit Agricole's CDS.
\begin{figure}[H]
    \centering
    \includegraphics[width=16.5cm, height=6.5cm]{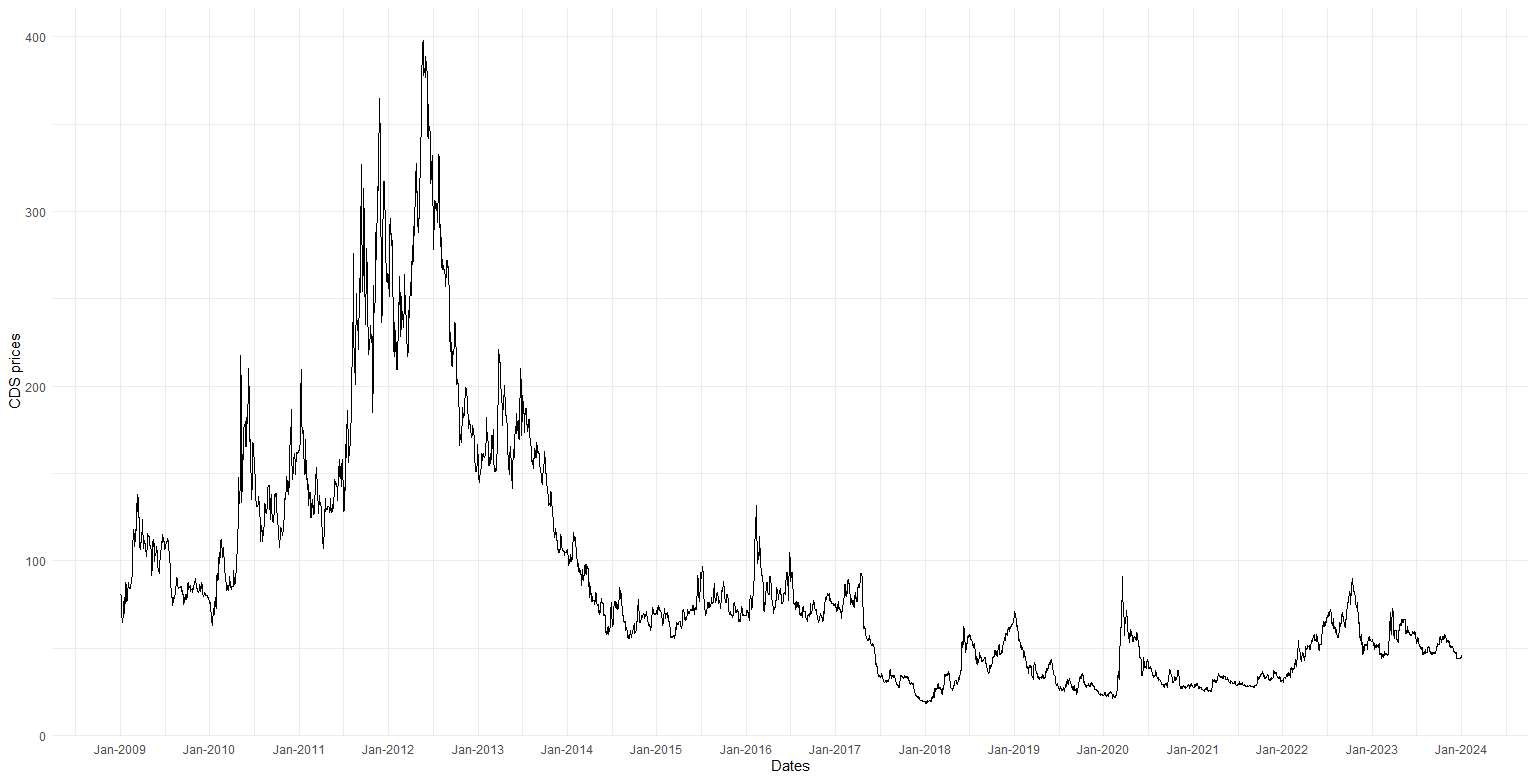}
    \caption{In BP, Credit Agricole's 5Y CDS}
    \label{fig:cds}
\end{figure}

Figures \ref{fig:di} and \ref{fig:srp} respectively show the bootsrapped default intensity and the survival probabilities.
\begin{figure}[H]
    \centering
    \includegraphics[width=15.5cm, height=6.cm]{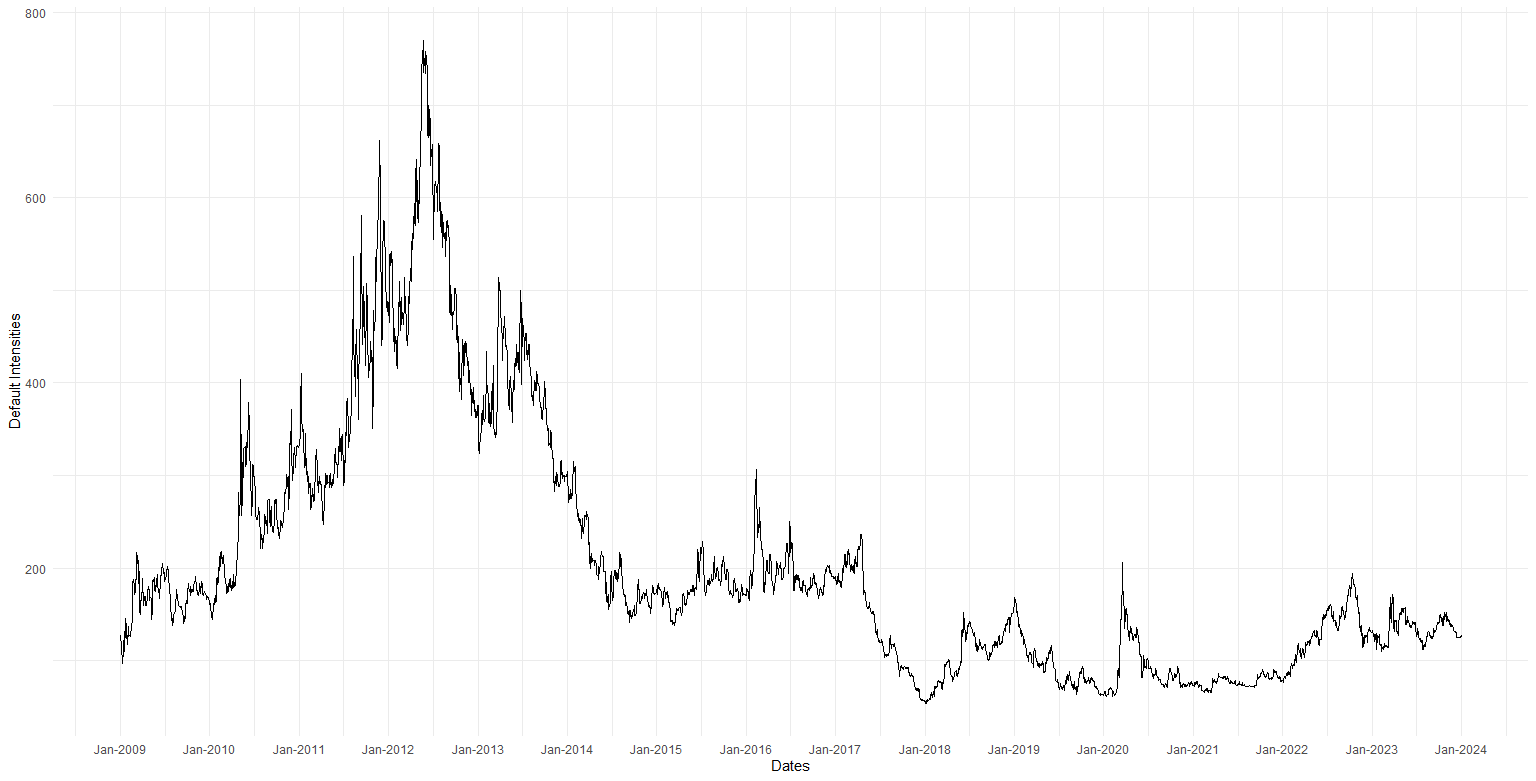}
    \caption{In BP, Credit Agricole's Default Intensity for the 5Y horizon}
    \label{fig:di}
\end{figure}

\begin{figure}[H]
    \centering
    \includegraphics[width=15.5cm, height=6.cm]{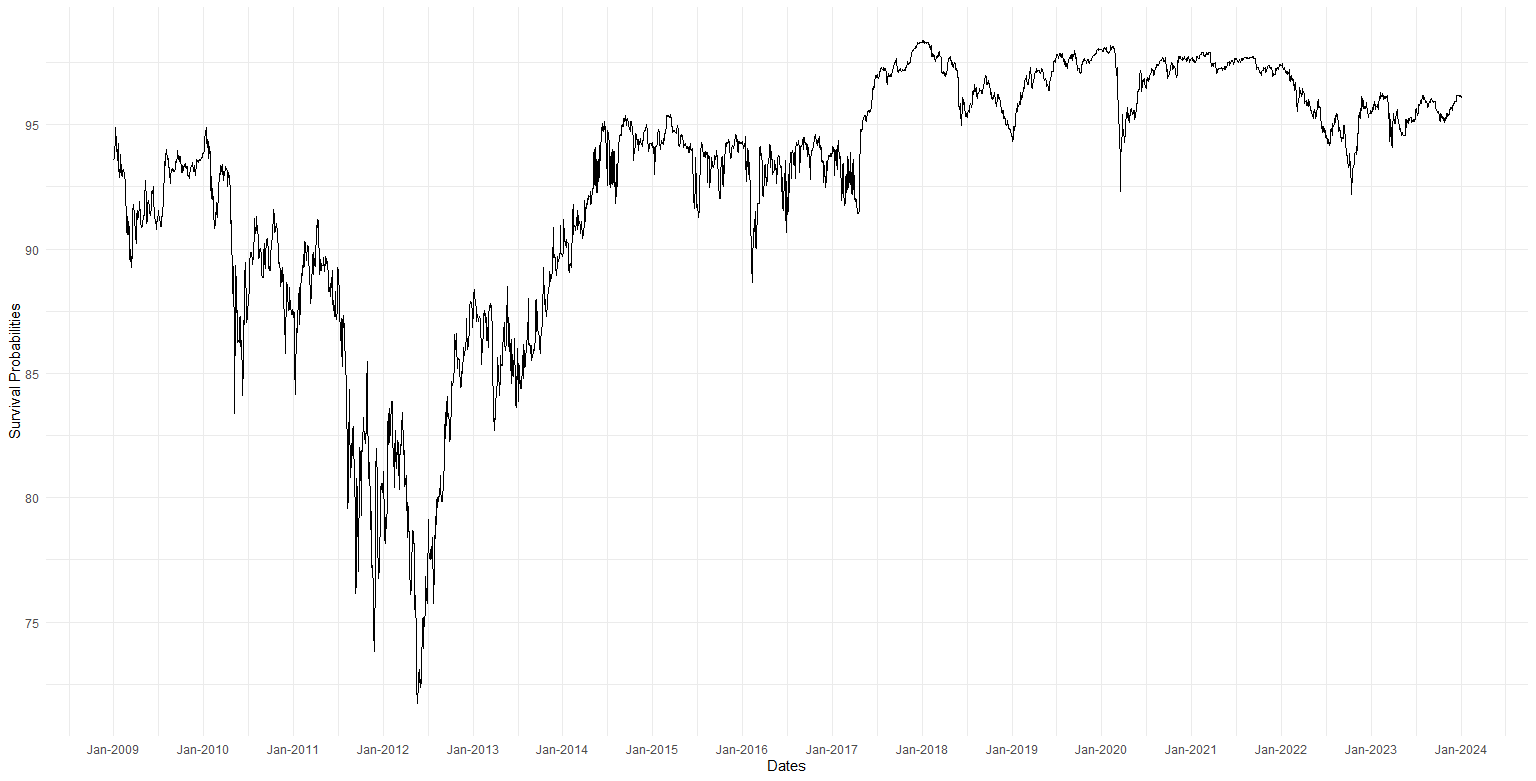}
    \caption{In percents, Credit Agricole's Survival Probabilities for the 5Y horizon}
    \label{fig:srp}
\end{figure}

\end{document}